\documentclass[12pt,draftcls,onecolumn,english]{IEEEtran}
\usepackage{amsmath}
\usepackage{subeqnarray}
\usepackage{cases}
\usepackage{framed}
\usepackage{amsfonts}
\usepackage{amssymb}
\usepackage{graphicx}
\usepackage{subfigure}
\usepackage{babel}
\usepackage{cite}
\usepackage{color}
\usepackage[nomarkers,nofiglist,notablist]{endfloat}%

\providecommand{\U}[1]{\protect\rule{.1in}{.1in}}

\newtheorem{remark}{Remark}

\newtheorem{proposition}{Proposition}

\begin{document}

\title{Transmission Schemes for Four-Way Relaying in Wireless Cellular Systems}
\author{%

\begin{tabular}
[c]{c}%
Huaping Liu$^{\dagger,\ast}$, Petar~Popovski$^{\ast}$, Elisabeth de Carvalho$^{\ast}$\\ Yuping Zhao$^{\dagger}$, Fan Sun$^{\ast}$ and Chan Dai Truyen Thai$^{\ast}$\\
$\dagger$State Key Laboratory of Advanced Optical Communication Systems and Networks\\ Peking University, China\\
$\ast$Department of Electronic Systems, Aalborg University, Denmark\\
Email:  $\left\{ {{\text{liuhp,yuping.zhao}}} \right\}$@pku.edu.cn, $\left\{ {{\text{petarp,edc,fs,ttc}}} \right\}$@es.aau.dk\\
\end{tabular}
}
\maketitle

\begin{abstract}
Two-way relaying in wireless systems has initiated a large research effort during the past few years. While one-way relay with a single data flow introduces loss in spectral efficiency due to its half-duplex operation, two-way relaying based on wireless network coding regains part of this loss by simultaneously processing the two data flows. In a broader perspective, the two-way traffic pattern is rather limited and it is of interest to investigate other traffic patterns where such a simultaneous processing of information flows can bring performance advantage. In this paper we consider a scenario beyond the usual two--way relaying:  \emph{a four-way relaying}, where each of the two Mobile Stations (MSs) has a two-way connection to the same Base Station (BS), while each connection is through a dedicated Relay Station (RS). While both RSs are in the range of the same BS, they are assumed to have \emph{antipodal positions} within the cell, such that they do not interfere with each other. We introduce and analyze a two-phase transmission scheme to serve the four-way traffic pattern defined in this scenario. Each phase consists of combined broadcast and multiple access. We analyze the achievable rate region of the new schemes for two different operational models for the RS, Decode-and-Forward (DF) and Amplify-and-Forward (AF), respectively. We compare the performance with a state-of-the-art reference scheme, time sharing is used between the two MSs, while each MS is served through a two-way relaying scheme. The results indicate that, when the RS operates in a DF mode, the achievable rate regions are significantly enlarged. On the other hand, for AF relaying, the gains are rather modest. The practical implication of the presented work is a novel insight on how to improve the spatial reuse in wireless cellular networks by coordinating the transmissions of the antipodal relays.
\end{abstract}

\section{Introduction}
Relay-based communication has matured in the past decades, both in terms of knowing the fundamental limits, but also with respect to practical system implementation and utilization. Capacity bounds and various cooperative strategies for relay networks have been studied in \cite{Cover}. \cite{Kramer} developed Decode-and-Forward (DF) relaying to multiple access relay channels and broadcast relay channels, and generalized Compress-and-Forward (CF) relaying to multiple relays as well. A paradigm shift occurred with the concept of network coding \cite{Ahlswede}, in which the relays process multiple data flows simultaneously and transmit functions of the incoming communication flows, rather than only replicating the incoming flows. This idea can bring profound gains in a wireless setting, notably in a scenario with two-way relaying\cite{Popovski}, \cite{Katti}. An overview of bidirectional relay protocols is given in \cite{PetarTWR}. The achievable rate regions for the two-way relaying channel under full-duplex assumption were given in \cite{Rankov}, but a viable assumption for wireless receivers is that they can operate in a half-duplex manner \cite{Madsen} \cite{Petarphy}. Note that for the half-duplex wireless transceivers, one-way relaying suffers from a loss in spectral efficiency as the relay cannot transmit and receive at the same time. Wireless network coding can help to regain part of this loss when two or more data flows are served simultaneously through the same relay. The work \cite{Kim} has analyzed the achievable rate regions of two-way relaying under half-duplex condition and compared various bidirectional relay protocols. Another related work is\cite{Oechtering}, which proves the optimal broadcast strategy for bidirectional relaying.
\begin{figure}
\centering
\includegraphics[width=10 cm]{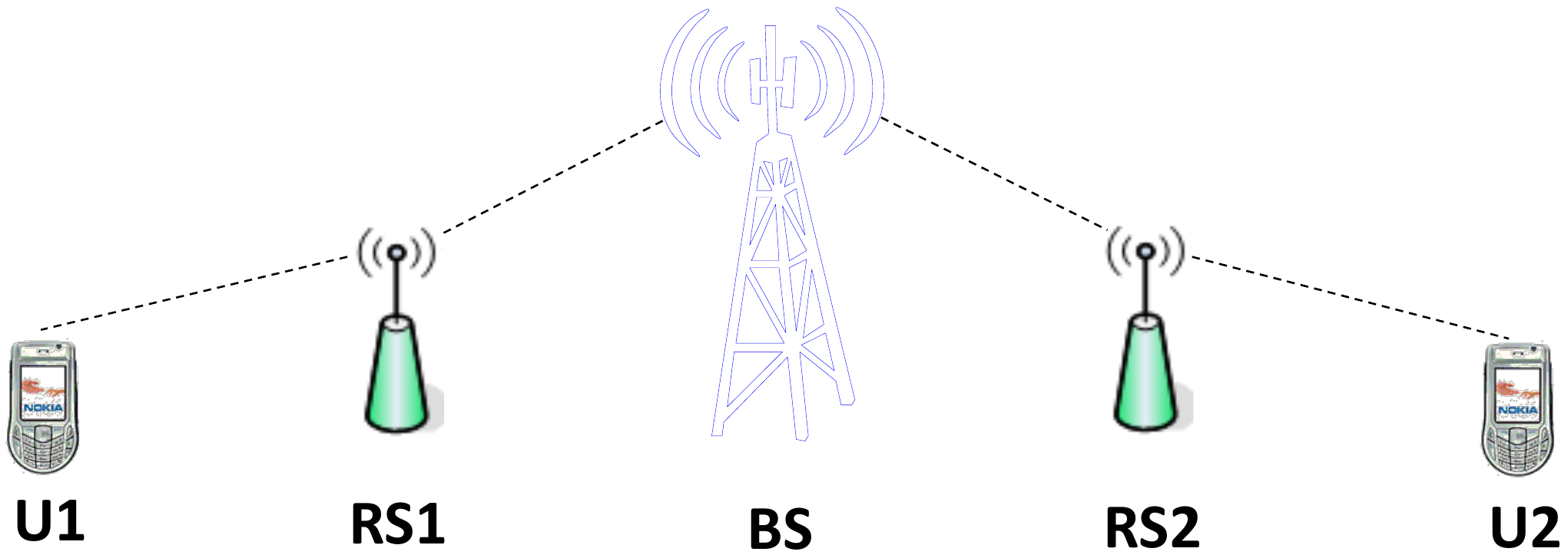}
\caption{A cellular network with four-way relaying using a Base Station (BS), two Relay Stations (RSs), denoted RS$1$ and RS$2$, and two Mobile Stations (MSs) U$1$ and U$2$.}%
\label{sen}%
\end{figure}

Although the benefits of wireless network coding have largely been confined to the canonical two-way relaying scenario, one can identify the underlying principles and then apply them in more generalized scenarios. Those principles are (1) simultaneous service of multiple flows over the wireless medium and (2) cancellation of interference based on previously gathered information. We have utilized these principles in order to devise transmission schemes in the case when the flows to a relayed and a direct user are simultaneously served~\cite{ChanICC11}. One of the main objectives of this paper is to leverage on these principles and investigate how they can be applied in a new, and very practical scenario. Namely, this involves  \emph{four-way relaying}, in which two Mobile Stations (MSs) have each a two-way connection to the same Base Station (BS), while each of them uses a dedicated Relay Station (RS). This is depicted on Fig.~\ref{sen}. Both RSs are in the range of the same BS, but they are at the ``opposite sides'' of the cell, i. e. have \emph{antipodal positions}, in a sense that they do not interfere with each other. For brevity, let BS$\rightarrow$RS$1$ denote the directional link from BS to RS$1$ (and similar for the other links). The state-of-the-art conventional scheme for this scenario is multiplexing two independent two-way relaying schemes in time as shown in Fig.~\ref{twoway:a}. Another conventional scheme is simultaneously conducting transmissions which do not interfere with each other as seen in Fig.~\ref{twoway:b}.

A similar scenario has been considered~\cite{Fang}, where DS-CDMA is used to avoid interference, the nodes use BPSK modulation and the relay applies  physical layer network coding (denoise-and-forward). Unlike \cite{Fang}, in our work we do not use orthogonal CDMA codes, but take advantage of the antipodal relay deployment in order to coordinate the interference. Furthermore, in our work we are not constrained to a specific modulation type, since we analyze the achievable rate region when the nodes use information-theoretic Gaussian codebooks. A careful look at the scenario reveals that the number of design possibilities is significantly increased. For example, an alternative scheme (Fig.~\ref{twoway:b}) could be the following 4-stage scheme: (1) BS$\rightarrow$RS$1$ and U$2$$\rightarrow$RS$2$, (2) RS$1$$\rightarrow$U$1$ and RS$2$$\rightarrow$BS, (3) BS$\rightarrow$RS$2$ and U$1$$\rightarrow$RS$1$, and (4) RS$1$$\rightarrow$BS and RS$2$$\rightarrow$U$2$. With similar observation, other transmission schemes can be proposed. We focus on a scheme
that consists of only two stages, where each phase consists of combined broadcast and multiple access. In the first phase, the BS broadcasts to RS$1$ and RS$2$ using superposition coding, while U$1$ and U$2$ are simultaneously carrying out their uplink transmissions. In the second phase, the relays RS$1$ and RS$2$ are broadcasting: RS$1$ to BS and U$1$, while RS$2$ to BS and U$2$. Assuming DF or AF operation at the relay, a reference scheme that represents the state-of-the-art can be defined as follows. Each MS is served through a two-way relaying scheme, while the BS applies time-sharing in order to serve both MSs. The achievable rate region for the considered scenario is four-dimensional. In order to obtain a better insight, we parametrize the two-way traffic associated with each of the MSs by defining \emph{downlink-uplink ratio}. The results show that, when the RS operates in a DF mode, the achievable rate regions are significantly enlarged. When AF is applied, in most cases the proposed two-phase transmission scheme has larger sum-rate than the reference scheme which has four phases.

The paper is organized as follows. Section (\ref{model}) introduces the notations and definitions about system and channel models.  Section (\ref{newscheme}) includes the analysis of the achievable rate region for two-phase transmission protocols.
Section (\ref{twowayscheme}) presents the achievable rate region of the reference transmission protocols which have four phases. Section (\ref{numerical}) shows the figures under some special conditions to give out an intuitive insight for the achievable rate regions.

\section{System and Channel Models}\label{model}
Consider a bidirectional cellular network in which a BS intends to exchange information with two mobile stations U$1$ and U$2$ with the aid of two relay stations RS$1$ and RS$2$, as Fig.~\ref{sen} shows. All
the nodes are half-duplex, such that a node can either transmit or receive at a given time. We assume that the MSs do not have direct links to/from the BS. For simplicity, we use $1,R1,B,R2,2$ as the indices in the formulas to denote the mobile station U$1$, relay station RS$1$, base station BS, relay station RS$2$ and mobile station U$2$ respectively. The channels are denoted by $h_{11}$(U$1$-RS$1$), $h_{12}$(RS$1$-BS), $h_{22}$(BS-RS$2$) and $h_{21}$(RS$2$-U$2$). Each channel ${h_l},l \in \{ 11,12,22,21\}$, is reciprocal, known at all the nodes. Each MS has a two-way, uplink/downlink traffic to/from the BS. The noise at all receivers ${z_j} \sim \mathcal{CN}(0,{\sigma ^2}) , j \in \{ 1,R1,B,R2,2\}$ is independent Additive White Gaussian Noise (AWGN) with zero mean and unit variance. If node $j$ is transmitting, its transmission power is bounded by ${\bar P_j}  , j \in \{ 1,R1,B,R2,2\}$, i.e., $E\left\{ {{{\left| {{x_j}} \right|}^2}} \right\} = {P_j} \le {\bar P_j}$ . The capacity of a single link is $C\left( \gamma  \right) = {\log _2}\left( {1 + \gamma } \right)$, where $\gamma$ is the Signal-to-Noise Ratio (SNR).

Fig.~\ref{twoway:a} illustrates a state-of-the-art transmission scheme based on time-division between two different two-way relaying instances. The first phase is the multiple access (MA) from U$1$ and BS to RS$1$. The second phase is the broadcast (BC) from RS$1$ to U$1$ and BS. Similarly, the third phase is the MA from U$2$ and BS to RS$2$, while the fourth phase is the BC from RS$2$ to U$2$ and BS. This four-phase four-way relaying scheme will serve as a reference scheme.
\begin{figure}[t]
  \centering
  \subfigure[Four stages scheme with 2 two-way relaying sessions]{
    \label{twoway:a} 
    \includegraphics[width=10 cm]{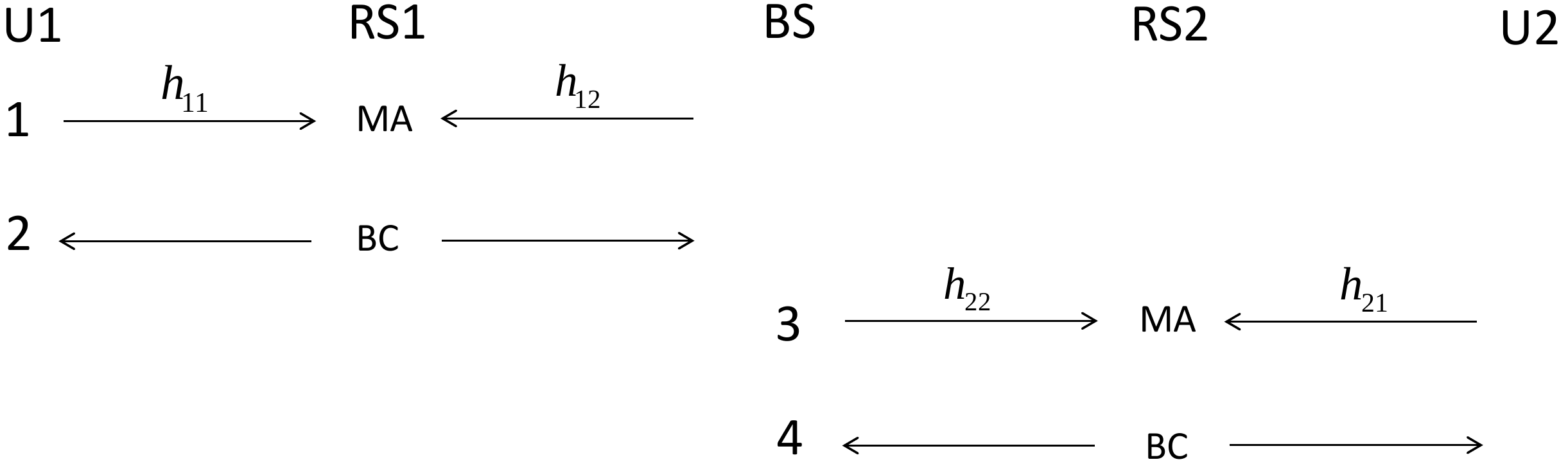}}
  \hspace{1in}
  \subfigure[Alternative four stages scheme]{
    \label{twoway:b} 
    \includegraphics[width=10 cm]{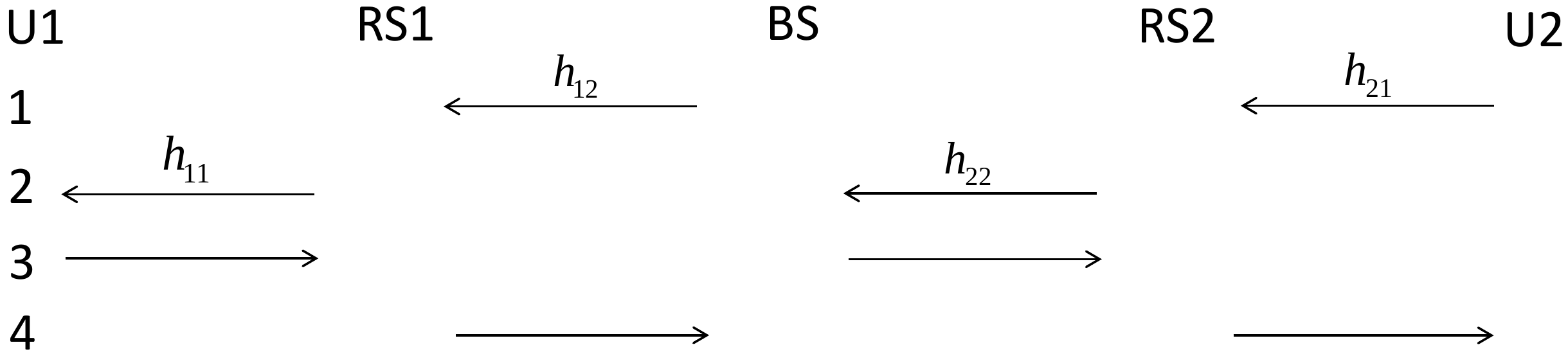}}
  \caption{Transmission schemes for four--way relaying consisting of four phases.}
  \label{twoway} 
\end{figure}

In the proposed scheme, the two downlink signals for U$1$ and U$2$ are broadcast by the BS using superposition coding. The whole transmission contains only two phases where communications with the two users occur simultaneously. Furthermore, a MA process and a BC process occur simultaneously in each phase, as the Fig.~\ref{4way} shows. The proposed scheme makes full use of the side information at U$1$, U$2$ and BS to cancel the self-interference.
\begin{figure}[h]
\centering
\includegraphics[width=10 cm]{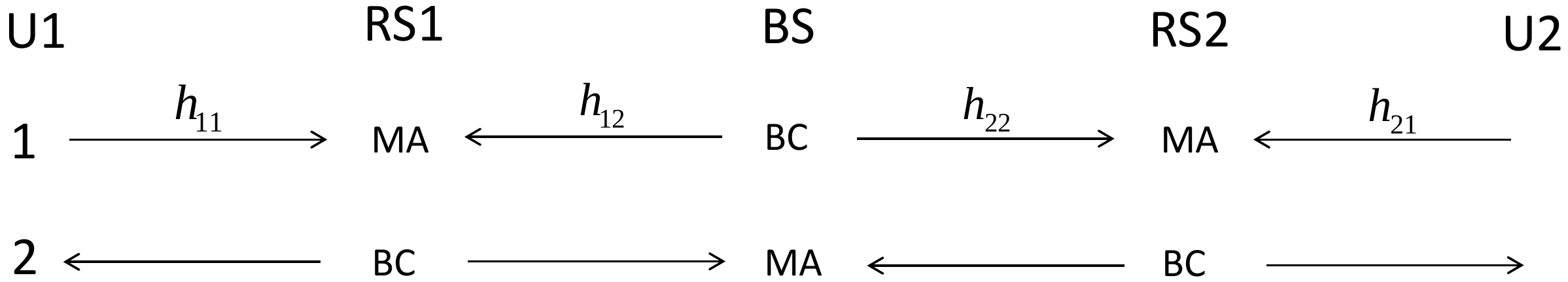}
\caption{A new transmission scheme for four--way relaying consisting of two phases.}%
\label{4way}%
\end{figure}
\section{Achievable rate regions of two-phase four-way relaying}\label{newscheme}
The signals sent by U$1$ and U$2$ are denoted by ${x_1}$ and ${x_2}$, respectively. The BS uses superposition coding and the broadcast signal is
\begin{equation}
{x_B} = \sqrt \alpha  {x_{B1}} + \sqrt {1 - \alpha } {x_{B2}},~\alpha  \in [0,1]
\label{supcode}%
\end{equation}where ${x_{B1}}({x_{B2}})$ is the signal intended for U$1$(U$2$) and $E\left\{ {{{\left| {{x_{B1}}} \right|}^2}} \right\} = E\left\{ {{{\left| {{x_{B2}}} \right|}^2}} \right\} = E\left\{ {{{\left| {{x_B}} \right|}^2}} \right\}={P_B} \le  {\bar P_B}$. Define $R_i^d$ as the downlink data rate of U$i$ and define $R_i^u$ as the uplink data rate of U$i$.

There are one BC process and two MA processes in the first phase: BS broadcasts the data intended to U$1$ and the data intended to U$2$ using superposition coding, while U$1$ and U$2$ transmit their uplink data. A MA occurs at RS$1$ with BS and U$1$ as transmitters, while another MA occurs at RS$2$ with BS and U$2$ as transmitters. In the second phase, there are two BC processes and one MA process. Each relay broadcasts a signal that is a function of the signal received in the first phase, while the BS acts as a receiver over a MA channel as the two RSs are transmitting simultaneously. Note that this is not an ordinary MA channel, since BS has a side information (i. e. its own information sent in phase 1) about the transmitted signals from the RSs. Finally, in the second phase, each of the users receives a
signal only from the relay and, similar to the usual two-way relaying, removes the self-information and decodes the desired signal. Next, we derive the rate region of the four-way relaying scheme when the relay operates in AF and DF mode, respectively.

\subsection{Amplify-and-forward}
At the end of phase 1, RS$i$ receives
\begin{equation}
{y_{Ri}} = {h_{i1}}{x_i} + {h_{i2}}{x_B} + {z_{Ri}},   i = 1,2.
\label{AF_step1_R1}%
\end{equation}

After reception, RS$1$ and RS$2$ amplify the received signals as follow:
\begin{equation}
{x_{Ri}} = {\beta _i}{y_{Ri}} , {\beta _i} = {{{P_{Ri}}} \over {\sqrt {{{\left| {{h_{i1}}} \right|}^2}{P_i} + {{\left| {{h_{i2}}} \right|}^2}{P_B} + 1} }},   i = 1,2
\label{AF_step2_R1}%
\end{equation}
here, ${\beta _i}$ is the amplification factor according to the relay transmission power constraints. ${x_{Ri}}$ is the signal broadcast by RS$i$.

At the end of phase 2, U$i$ receives ${y_i}$. Based on the side information about $x_i$, the channels and ${\beta _i}$, U$i$ can cancel the contribution of ${x_i}$ from ${y_i}$ to get ${\bar y_i}$,
\begin{equation}
{{\bar y}_i} = {h_{i1}}{h_{i2}}{\beta _i}\left( {\sqrt \alpha  {x_{B1}} + \sqrt {1 - \alpha } {x_{B2}}} \right)  + {h_{i1}}{\beta _i}{z_{Ri}} + {z_i},   i = 1,2.
\label{AF_step2_ui}%
\end{equation}

 As the BS has the side information about ${x_B}$, it can remove the contribution of ${x_B}$ from the received signal ${y_B}$ to get ${{\bar y}_B}$,
\begin{equation}
{{\bar y}_B} = {h_{11}}{h_{12}}{\beta _1}{x_1} + {h_{22}}{h_{21}}{\beta _2}{x_2} + {h_{12}}{\beta _1}{z_{R1}} + {h_{22}}{\beta _2}{z_{R2}} + {z_B}.
\label{AF_step2_BS}%
\end{equation}

Equation (\ref{AF_step2_ui}) describes a BC channel where the signals are sent through superposition coding: ${x_{Bi}}$ is intended to user i.  Equation (\ref{AF_step2_BS}) describes a MA channel. Using side information, the system becomes equivalent to 2 direct communication channels between the BS and the users. The first channel is a BC channel and the second channel is a MA channel, for which known information theory results can be used.

The SNRs of an equivalent BC channel defined in (\ref{AF_step2_ui}) are given by
\begin{align}
{S_i}  &= \frac{{E\left\{{{\left| {{h_{i1}}{h_{i2}}{\beta _i}{x_B}} \right|}^2}\right\}}}{{E\left\{ {{\left| {{h_{i1}}{\beta _i}{z_{Ri}} + {z_i}} \right|}^2}\right\}}}
={{{{\left| {{h_{i1}}} \right|}^2}{{\left| {{h_{i2}}} \right|}^2}{P_B}{P_{Ri}}} \over {{{\left| {{h_{i1}}} \right|}^2}\left( {{P_i} + {P_{Ri}}} \right) + {{\left| {{h_{i2}}} \right|}^2}{P_B} + 1}} ,   i = 1,2.
 \label{Si}
\end{align}

The SNRs of an equivalent MA channel defined in (\ref{AF_step2_BS}) are given by
\begin{equation}
\begin{gathered}
  {S'_i} = \frac{{E\left\{{{\left| {{h_{i1}}{h_{i2}}{\beta _i}{x_i}} \right|}^2}\right\}}}{{E\left\{{{\left| {{h_{12}}{\beta _1}{z_{R1}} + {h_{22}}{\beta _2}{z_{R2}} + {z_B}} \right|}^2}\right\}}}
      ={{{{{{\left| {{h_{i1}}} \right|}^2}{{\left| {{h_{i2}}} \right|}^2}{P_i}{P_{Ri}}} \over {{{\left| {{h_{i1}}} \right|}^2}{P_i} + {{\left| {{h_{i2}}} \right|}^2}{P_B} + 1}}} \over {{{{{\left| {{h_{12}}} \right|}^2}{P_{R1}}} \over {{{\left| {{h_{11}}} \right|}^2}{P_1} + {{\left| {{h_{12}}} \right|}^2}{P_B} + 1}} + {{{{\left| {{h_{22}}} \right|}^2}{P_{R2}}} \over {{{\left| {{h_{21}}} \right|}^2}{P_2} + {{\left| {{h_{22}}} \right|}^2}{P_B} + 1}} + 1}} , i = 1,2.  \hfill \\
\end{gathered}
\label{AF_Sp}%
\end{equation}

\begin{proposition} The achievable rate region of the two-phase four-way relaying scheme using AF is
\begin{eqnarray}
\nonumber &&{\rm{when}}~~{S_1} > {S_2}\\
\label{newAF1} \nonumber &&~~~~R_1^d < \frac{1}{2}C\left( {{S_1}\alpha } \right), R_2^d < \frac{1}{2}C\left( {\frac{{{S_2}(1 - \alpha )}}{{{S_2}\alpha  + 1}}} \right), R_1^u < {1 \over 2}C\left( {{S'_1}} \right)\\
\label{newAF3}  &&~~~~R_2^u < {1 \over 2}C\left( {{S'_2}} \right), R_1^u + R_2^u < \frac{1}{2}C\left( {{S'_1} + {S'_2}} \right)\\
\nonumber &&{\rm{when}}~~{S_1} \leq {S_2}\\
\label{newAF21} \nonumber &&~~~~R_1^d < \frac{1}{2}C\left( {\frac{{{S_1}\alpha }}{{{S_1}(1 - \alpha ) + 1}}} \right), R_2^d < \frac{1}{2}C\left( {{S_2}(1 - \alpha )} \right), R_1^u < {1 \over 2}C\left( {{S'_1}} \right)\\
\label{newAF23}  &&~~~~R_2^u < {1 \over 2}C\left( {{S'_2}} \right), R_1^u + R_2^u < \frac{1}{2}C\left( {{S'_1} + {S'_2}} \right).
\end{eqnarray}
\end{proposition}

\begin{proof}
The whole system is composed by a BC channel and a MA channel.
\subsubsection{Broadcast Channel}
From \cite{Gamal}, when ${S_1} > {S_2}$ the optimal decoding consists in decoding ${x_{B2}}$ at U$2$ treating ${x_{B1}}$ as Gaussian noise. At U$1$, the decoder first decodes ${x_{B2}}$, then cancels ${x_{B2}}$ from the received signal ${{\bar y}_1}$. Finally, U$1$ decodes ${x_{B1}}$. The capacity region of the AWGN BC channel in (\ref{AF_step2_ui}) is
\begin{equation}
\label{AF_BC_case1}
R_1^d < \frac{1}{2}C\left( {{S_1}\alpha } \right), ~~R_2^d < \frac{1}{2}C\left( {\frac{{{S_2}(1 - \alpha )}}{{{S_2}\alpha  + 1}}} \right)
\end{equation}here, ${R_1^d}$ is the rate of ${x_{B1}}$, ${R_2^d}$ is the rate of ${x_{B2}}$. When ${S_1} \leq {S_2}$ the capacity region is similar. The scaling factors 1/2 in (\ref{AF_BC_case1}), account for the half-duplex constraint. Furthermore, if RS$1$ and RS$2$ use AF, phase 1 and phase 2 have the same duration. Then U$1$ and U$2$ only spend half of the time to receive the signal. BS spends half of the time to transmit. When ${S_1} \leqslant {S_2}$, the capacity region of the AWGN BC channel in (\ref{AF_step2_ui}) is,
\begin{equation}
\label{AF_BC_3}  R_1^d < \frac{1}{2}C\left( {\frac{{{S_1}\alpha }}{{{S_1}(1 - \alpha ) + 1}}} \right),~~R_2^d < \frac{1}{2}C\left( {{S_2}(1 - \alpha )} \right).
\end{equation}

\subsubsection{Multiple access channel}
From \cite{Gamal}, the capacity region of AWGN MA channel in (\ref{AF_step2_BS}) is
\begin{equation}
\label{AF_MACC}
R_1^u < \frac{1}{2}C\left( {{S'_1}} \right),~~R_2^u < \frac{1}{2}C\left( {{S'_2}} \right),~~ R_1^u + R_2^u < \frac{1}{2}C\left( {{S'_1} + {S'_2}} \right).
\end{equation} We have scaling factors 1/2 in (\ref{AF_MACC}), because U$1$ and U$2$ only spend half of the time to transmit the signal. BS spends half of the time to receive the signal.



Combining (\ref{AF_BC_case1})-(\ref{AF_MACC}), we obtain the achievable rate region of two-phase AF scheme as shown in proposition 1.
\end{proof}

\subsection{Decode-and-forward}
At the end of phase 1, RS$i$ receives
\begin{equation}
{y_{Ri}} = {h_{i1}}{x_i} + {h_{i2}}\sqrt \alpha_i  {x_{Bi}} + {h_{i2}}\sqrt {1 - \alpha_i } {x_{Bj}} + {z_{Ri}}
\label{y_R1}.%
\end{equation}
here $\left( {i,j} \right) \in \left\{ {\left( {1,2} \right),\left( {2,1} \right)} \right\}$, ${\alpha _1} = \alpha ,{\alpha _2} = 1 - \alpha $.

The decoder in RS$i$ decodes the signal ${x_i}$ and ${x_{Bi}}$ , then re-encodes the messages of ${x_i}$ and ${x_{Bi}}$ into $x_{Ri}$. RS$i$ does not need to decode ${x_{Bj}}$, as it is intended to U$j$. In fact, ${x_{Bj}}$ should only be decoded if there is a  benefit for decoding ${x_i}$ and ${x_{Bi}}$, otherwise ${x_{Bj}}$ should be treated as noise. However, the successful decoding of ${x_{Bj}}$ will impose a rate limitation on $R_j^d$.

At the end of phase 2, U$i$ receives
\begin{equation}
y_i = {h_{i1}}x_{Ri} + {z_i}.
\label{yi}%
\end{equation}

Meanwhile BS receives,
\begin{equation}
y_B = {h_{12}}x_{R1} + {h_{22}}x_{R2} + {z_B}.
\label{ybs}%
\end{equation}

Equations (\ref{yi}) and (\ref{ybs}) describe the BC channel from RS$i$ to U$i$ and BS which have the side information. Equation (\ref{ybs}) describes the MA channel at BS with side information.

\begin{proposition}The achievable rate region of the two-phase four-way relaying scheme using DF is the convex closure of all 4-dimensional rate tuples satisfying ${\left( {R_1^u,R_1^d,R_2^u,R_2^d} \right) \in D}$.
\begin{subequations}
\label{defineD}
\begin{align}
\label{defineD3} &D_{Ri}^{\left( 1 \right)} = M_{Ri}^2 \cup M_{Ri}^3\\
\label{defineD2} &{D^{\left( 1 \right)}} = D_{R1}^{\left( 1 \right)} \cap D_{R2}^{\left( 1 \right)}\\
\label{defineD1} &D = {D^{\left( 1 \right)}} \cap {D^{\left( 2 \right)}}
\end{align}
\end{subequations}where $M_{Ri}^2$, $M_{Ri}^3$ and ${D^{\left( 2 \right)}}$ are defined as follows
\begin{eqnarray}
M_{Ri}^2 &=& \left\{ {\left( {R_i^u,R_i^d,R_j^u,R_j^d} \right)|\left( {R_i^u,R_i^d} \right) \in } \right (\ref{MAC_R1M1})\} \textrm{ where }\left( {i,j} \right) \in \left\{ {\left( {1,2} \right),\left( {2,1} \right)} \right\} \label{M^2} \\
M_{Ri}^3 &=& \left\{ {\left( {R_i^u,R_i^d,R_j^u,R_j^d} \right)|\left( {R_i^u,R_i^d,R_j^d} \right) \in } \right (\ref{MAC_R1})\}\textrm{ where } \left( {i,j} \right) \in \left\{ {\left( {1,2} \right),\left( {2,1} \right)} \right\} \label{M^3} \\
D^{\left( 2 \right)} &=& \left\{ \left( {R_1^u,R_1^d,R_2^u,R_2^d} \right)|\left( {R_1^u,R_1^d,R_2^u,R_2^d} \right) \in  \right (\ref{bcop}) \} \label{D^2}
\end{eqnarray}
where $0 \leq \alpha_i, \tau \leq 1$.
\end{proposition}

\begin{subequations}
\label{MAC_R1M1}
\begin{numcases}{}
\label{MAC_R1_s21}  R_i^u < \tau C\left( {{{{\left| {{h_{i1}}} \right|}^2}{P_i}} \over {\left( {1 - {\alpha _i}} \right){{\left| {{h_{i2}}} \right|}^2}{P_B} + 1}} \right),~~R_i^d < \tau C\left( {{{\alpha _i}{{\left| {{h_{i2}}} \right|}^2}{P_B}} \over {\left( {1 - {\alpha _i}} \right){{\left| {{h_{i2}}} \right|}^2}{P_B} + 1}} \right)\\
\label{MAC_R1_s22}  R_i^u + R_i^d < \tau C\left({{{{\left| {{h_{i1}}} \right|}^2}{P_i} + {\alpha _i}{{\left| {{h_{i2}}} \right|}^2}{P_B}} \over {\left( {1 - {\alpha _i}} \right){{\left| {{h_{i2}}} \right|}^2}{P_B} + 1}} \right).
\end{numcases}
\end{subequations}

\begin{subequations}
\label{MAC_R1}
\begin{numcases}{}
\label{MAC_R1_1}  R_i^u < \tau C\left( {{{\left| {{h_{i1}}} \right|}^2}{P_i}} \right),~~R_i^d < \tau C\left({{\alpha _i}{{\left| {{h_{i2}}} \right|}^2}{P_B}}\right),~~R_j^d < \tau C\left( {\left( {1 - {\alpha _i}} \right){{\left| {{h_{i2}}} \right|}^2}{P_B}} \right)\\
\label{MAC_R1_4}  R_i^u + R_i^d < \tau C\left( {{{\left| {{h_{i1}}} \right|}^2}{P_i}} + {{\alpha _i}{{\left| {{h_{i2}}} \right|}^2}{P_B}} \right)\\
\label{MAC_R1_5}  R_i^u + R_j^d < \tau C\left({{{\left| {{h_{i1}}} \right|}^2}{P_i}} +{\left( {1 - {\alpha _i}} \right){{\left| {{h_{i2}}} \right|}^2}{P_B}} \right)\\
\label{MAC_R1_6}  R_i^d + R_j^d < \tau C\left({{\left| {{h_{i2}}} \right|}^2}{P_B}\right)\\
\label{MAC_R1_7}  R_i^u + R_i^d + R_j^d < \tau C\left({{{\left| {{h_{i1}}} \right|}^2}{P_i}} + {{\left| {{h_{i2}}} \right|}^2}{P_B}\right).
\end{numcases}
\end{subequations}

\begin{subequations}
\label{bcop}
\begin{numcases}{}
\label{bcop_1}  R_1^d < \left( {1 - \tau } \right)C\left( {\left| {{h_{11}}} \right|^2}{P_{R1}} \right),~~R_2^d < \left( {1 - \tau } \right)C\left( {\left| {{h_{21}}} \right|^2}{P_{R2}} \right)\\
\label{bcop_3}  R_1^u < \left( {1 - \tau } \right)C\left( {\left| {{h_{12}}} \right|^2}{P_{R1}} \right),~~R_2^u < \left( {1 - \tau } \right)C\left( {\left| {{h_{22}}} \right|^2}{P_{R2}} \right)\\
\label{bcop_5}  R_1^u + R_2^u < \left( {1 - \tau } \right)C\left({\left| {{h_{12}}} \right|^2}{P_{R1}} + {\left| {{h_{22}}} \right|^2}{P_{R2}} \right)
\end{numcases}
\end{subequations}

\medskip

\begin{proof}
Let us focus on the communication of U$1$ in phase 1; the other user is treated similarly. In (\ref{M^2}) and (\ref{M^3}) we set $i=1$ and $j=2$. U$1$ does not need to receive $x_{B2}$ and therefore the relay RS$1$ will only decode $x_{B2}$ if it helps to obtain higher rates $(R_1^u, R_1^d)$ for the data flows supported through the RS$1$. $M_{R1}^2$ is the rate region that is obtained when RS$1$ receives $x_1$ and $x_{B1}$ over a MA channel, while treating $x_{B2}$ as noise. Note that $M_{R1}^2$ is a four-dimensional region in which no constraints are put on $(R_2^u, R_2^d)$, which means that they can have arbitrary non-negative values. This is made for the sake of consistent notation; the actual  upper limits on $(R_2^u, R_2^d)$ will be set by intersecting the regions in (\ref{defineD}). If we put the constraint that RS$1$ should decode $x_{B2}$, then we need to consider a three-user MA channel at RS$1$ with the signals $x_1, x_{B1}, x_{B2}$. The 4-dimensional rate region is given by (\ref{M^3}) where $R_2^u$ is unconstrained.
The union $M_{Ri}^2 \cup M_{Ri}^3$ is a four-dimensional region $D_{R1}^{\left( 1 \right)}$. The projection of $D_{R1}^{\left( 1 \right)}$ on the plane $(R_1^u, R_1^d)$ defines all possible rate pairs that can be decoded at the RS1. Clearly, some rate pairs are decodable when $x_{B2}$ is decoded, others when it is treated as noise. Using the similar analysis, we obtain the MA region at RS$2$ in phase 1 which is $D_{R2}^{\left( 1 \right)}$.

The 4-dimensional rate region that describes phase 1, including the operation at both RS1 and RS2, is denoted by ${D^{\left( 1 \right)}}$. Note that, for some points in $D_{R1}^{\left( 1 \right)}$, the values of $R_2^u$ or $R_2^d$ can be arbitrarily large; and the same is valid for the values of $R_1^u$ or $R_1^d$ in $D_{R2}^{\left( 1 \right)}$. Nevertheless, due to intersection, ${D^{\left( 1 \right)}}$ is set of non-negative coordinate points in which each coordinate has an upper bound.

We use  ${D^{\left( 2 \right)}}$ to denote 4-dimensional rate region that is achievable during the second phase. The proof of ${D^{\left( 2 \right)}}$ is provided in Appendix A. The rate region
$D$ is the intersection of the 4-dimensional rate region in phase 1 and phase 2, as indicated by (\ref{defineD1}). Note that the region coming from the intersection of different convex regions is also convex. However the the region coming from the union of different convex regions may not be convex, such that $D$ may be non-convex. Therefore, the achievable rate region is the convex closure of $D$.
\end{proof}

\begin{remark}
In phase 2 there are two BC processes from RS$1$ and RS$2$. The BC destinations are U$1$, U$2$ and BS which all have  side information. These two BC processes intersect at BS and become a MA process to the BS. A similar scenario is introduced in \cite{Oechtering} and \cite{KramerBC} which give an optimal BC strategy to two terminals where side information is available. However, \cite{Oechtering} and \cite{KramerBC} do not deal with a MA channel with side information and, to the best of our knowledge, there is no proved optimal BC strategy for this scenario. Phase 2 of the transmission scheme in this paper is the MA extension of the schemes in \cite{Oechtering} and \cite{KramerBC}.
\end{remark}

\section{Achievable rate regions of four-phase four-way relaying}\label{twowayscheme}
In this paper, we view the four-phase relaying scheme as the reference scheme. This is depicted on Fig.~\ref{twoway:a}, where the four-way transmission is decomposed into two time-multiplexed two-way relaying transmissions.

As will be shown later, the expression of the achievable rate regions of this four-phase relaying scheme are independent from the time-sharing fraction (time-ratio) between the 2 two-way relay processes, for AF, and additionally for DF, the BC and MA phases. Unlike the two-phase relaying scheme, the time-ratio is an implicit variable and the achievable rate region can be obtained in closed-form.
\subsection{Amplify-and-forward}
\begin{proposition}The achievable rate region of the four-phase four-way relaying scheme using AF is characterized by ${l_1} + {l_2} < 1$, where ${l_1}$ and ${l_2}$ are defined as follow£º
\begin{subequations}
\label{L}
\begin{align}
\label{L1} &{l_1} = \max \left\{ {{{2R_1^u} \over {C\left({{{{\left| {{h_{11}}} \right|}^2}{{\left| {{h_{12}}} \right|}^2}{P_1}{P_{R1}}} \over {{{\left| {{h_{11}}} \right|}^2}{P_1} + {{\left| {{h_{12}}} \right|}^2}\left( {{P_{R1}} + {P_B}} \right) + 1}} \right)}},{{2R_1^d} \over {C\left({{{{\left| {{h_{11}}} \right|}^2}{{\left| {{h_{12}}} \right|}^2}{P_B}{P_{R1}}} \over {{{\left| {{h_{12}}} \right|}^2}{P_B} + {{\left| {{h_{11}}} \right|}^2}\left( {{P_{R1}} + {P_1}} \right) + 1}} \right)}}} \right\}\\
\label{L2} &{l_2} = \max \left\{ {{{2R_2^u} \over {C\left({{{{\left| {{h_{21}}} \right|}^2}{{\left| {{h_{22}}} \right|}^2}{P_2}{P_{R2}}} \over {{{\left| {{h_{21}}} \right|}^2}{P_2} + {{\left| {{h_{22}}} \right|}^2}\left( {{P_{R2}} + {P_B}} \right) + 1}} \right) }},{{2R_2^d} \over {C\left({{{{\left| {{h_{21}}} \right|}^2}{{\left| {{h_{22}}} \right|}^2}{P_B}{P_{R2}}} \over {{{\left| {{h_{22}}} \right|}^2}{P_B} + {{\left| {{h_{21}}} \right|}^2}\left( {{P_{R2}} + {P_2}} \right) + 1}}\right)}}} \right\}.
\end{align}
\end{subequations}

\end{proposition}

\begin{proof}
If the relay stations RS$1$ and RS$2$ use AF, then the duration of phase 1 (3) is equal to the duration of phase 2 (4).  Assuming that the total duration of all the phases is equal to unity, then phase 1 and phase 2 together occupy the time ratio $\eta$, while phase 3 and phase 4 is ${1{\rm{ - }}\eta }$. From \cite{Kim}, we get the constraints for phase 1 and phase 2
\setcounter{equation}{25}
\begin{subequations}
\label{AF_fours1}
\begin{numcases}{}
\label{AF_four1}  R_1^u < \frac{1}{2}\eta C\left({{{{\left| {{h_{11}}} \right|}^2}{{\left| {{h_{12}}} \right|}^2}{P_1}{P_{R1}}} \over {{{\left| {{h_{11}}} \right|}^2}{P_1} + {{\left| {{h_{12}}} \right|}^2}\left( {{P_{R1}} + {P_B}} \right) + 1}} \right)\\
\label{AF_four2}  R_1^d < \frac{1}{2}\eta C\left({{{{\left| {{h_{11}}} \right|}^2}{{\left| {{h_{12}}} \right|}^2}{P_B}{P_{R1}}} \over {{{\left| {{h_{12}}} \right|}^2}{P_B} + {{\left| {{h_{11}}} \right|}^2}\left( {{P_{R1}} + {P_1}} \right) + 1}} \right).
\end{numcases}
\end{subequations}
Similarly, the constraints for phase 3 and phase 4 can be obtained as
\begin{subequations}
\label{AF_fours2}
\begin{numcases}{}
  \label{AF_four3} R_2^u < \frac{{1{\rm{ - }}\eta }}{2}C\left({{{{\left| {{h_{21}}} \right|}^2}{{\left| {{h_{22}}} \right|}^2}{P_2}{P_{R2}}} \over {{{\left| {{h_{21}}} \right|}^2}{P_2} + {{\left| {{h_{22}}} \right|}^2}\left( {{P_{R2}} + {P_B}} \right) + 1}} \right)\\
\label{AF_four4}  R_2^d < \frac{{1{\rm{ - }}\eta }}{2}C\left({{{{\left| {{h_{21}}} \right|}^2}{{\left| {{h_{22}}} \right|}^2}{P_B}{P_{R2}}} \over {{{\left| {{h_{22}}} \right|}^2}{P_B} + {{\left| {{h_{21}}} \right|}^2}\left( {{P_{R2}} + {P_2}} \right) + 1}}\right).
\end{numcases}
\end{subequations}


Using the definition of $l_1$ and $l_2$ in (\ref{L1}) and (\ref{L2}), we can write (\ref{AF_fours1}) and (\ref{AF_fours2}) as
\begin{equation}
{l_1} < \eta, ~~~~{l_2} < 1{\rm{ - }}\eta.
\label{regionl}
\end{equation}

The inequalities in (\ref{regionl}) are equivalent to ${l_1} + {l_2} < 1$ which is the closed-form expression of the achievable rate region for two-way AF relaying scheme.
\end{proof}

\subsection{Decode-and-forward}
\begin{proposition}The achievable rate region of the four-phase four-way relaying scheme using DF is characterized by ${k_1} + {k_2}+{k_3} + {k_4} < 1$.
\begin{subequations}
\label{K}
\begin{align}
\label{k1} &{k_1} = \max \left\{ {{{R_1^u} \over {C\left({\left| {{h_{11}}} \right|^2}{P_1}\right)}},{{R_1^d} \over {C\left({\left| {{h_{12}}} \right|^2}{P_B}\right)}},{{R_1^u + R_1^d} \over {C\left({\left| {{h_{11}}} \right|^2}{P_1} + {\left| {{h_{12}}} \right|^2}{P_B}\right)}}} \right\} \\
\label{k2} &{k_2} = \max \left\{ {{{R_1^d} \over {C\left({\left| {{h_{11}}} \right|^2}{P_{R1}}\right)}},{{R_1^u} \over {C\left({\left| {{h_{12}}} \right|^2}{P_{R1}}\right)}}} \right\}\\
\label{k3} &{k_3} = \max \left\{ {{{R_2^u} \over {C\left({\left| {{h_{21}}} \right|^2}{P_2}\right)}},{{R_2^d} \over {C\left({\left| {{h_{22}}} \right|^2}{P_B}\right)}},{{R_1^u + R_1^d} \over {C\left({\left| {{h_{21}}} \right|^2}{P_2} + {\left| {{h_{22}}} \right|^2}{P_B}\right)}}} \right\}\\
\label{k4} &{k_4} = \max \left\{ {{{R_2^d} \over {C\left({\left| {{h_{21}}} \right|^2}{P_{R2}}\right)}},{{R_2^u} \over {C\left({\left| {{h_{22}}} \right|^2}{P_{R2}} \right)}}} \right\}.
\end{align}
\end{subequations}

\end{proposition}

\begin{proof}
When the two-way DF relaying scheme is used in the four-way relay cellular networks, the time duration of the 4 phases can be different from each other. We use ${\tau _1},{\tau _2},{\tau _3},{\tau _4}$ to denote the time ratios spent by phase 1,2,3,4 respectively. Notice that, ${\tau _1} +  {\tau _2} + {\tau _3} + {\tau _4} = 1$ must be satisfied. From \cite{Gamal}, we get the constraints on the rates for the MA process in phase 1,
\begin{equation}
\label{DF_four1}
R_1^u < {\tau _1}C\left({\left| {{h_{11}}} \right|^2}{P_1}\right),~~ R_1^d < {\tau _1}C\left({\left| {{h_{12}}} \right|^2}{P_B}\right),~~  R_1^u + R_1^d < {\tau _1}C\left({\left| {{h_{11}}} \right|^2}{P_1} + {\left| {{h_{12}}} \right|^2}{P_B}\right)
\end{equation}
The second phase is the BC process with side information to U$1$ and BS as U$1$ and BS both know what they sent in phase 1. From \cite{Oechtering}, the constraints on the rates for the broadcast process from RS$1$ can be obtained as
\begin{equation}
\label{DF_four2}
R_1^d < {\tau _2}C\left({\left| {{h_{11}}} \right|^2}{P_{R1}}\right),~~R_1^u < {\tau _2}C\left({\left| {{h_{12}}} \right|^2}{P_{R1}}\right).
\end{equation}

Similarly, we can obtain the constraints for the MA and BC process at RS$2$ in phase 3 and phase 4,
\vspace{-1pt}
\begin{subequations}
\label{DF_four4}
\begin{numcases}{}
\label{DF_four3_1}  R_2^u < {\tau _3}C\left({\left| {{h_{21}}} \right|^2}{P_2}\right),~~R_2^d < {\tau _3}C\left({\left| {{h_{22}}} \right|^2}{P_B}\right)\\
\label{DF_four3_3}  R_2^u + R_2^d < {\tau _3}C\left({\left| {{h_{21}}} \right|^2}{P_2} + {\left| {{h_{22}}} \right|^2}{P_B}\right)\\
\label{DF_four4_1}  R_2^d < {\tau _4}C\left({\left| {{h_{21}}} \right|^2}{P_{R2}}\right),~~R_2^u < {\tau _4}C\left({\left| {{h_{22}}} \right|^2}{P_{R2}} \right).
\end{numcases}
\end{subequations}

Applying (\ref{K}) into (\ref{DF_four1}), (\ref{DF_four2}) and (\ref{DF_four4}), we get
\begin{equation}
{k_i} < {\tau _i}, i=1,2,3,4
\label{tao}.%
\end{equation}

Incorporating ${\tau _1} +  {\tau _2} + {\tau _3} + {\tau _4} = 1$ into (\ref{tao}), we obtain ${k_1} + {k_2}+{k_3} + {k_4} < 1$ which is the closed-form expression of the achievable rate region for two-way DF relaying scheme.
\end{proof}

\section{Achievable rates with fixed Downlink-Uplink rate ratio}
\label{num}

The achievable rate regions described in section \ref{newscheme} and \ref{twowayscheme} have four dimensions. In order to get a better insight into the achievable rates, we impose a ratio between the uplink and downlink rates. It is practical to fix the downlink-uplink rate ratio for a given type of application e. g. gaming and calls have ratio of 1:1, web browsing has a ratio of about 5:1 \cite{Truong}. We assume that, for the $i-$th user, the downlink rate demand $R_i^d$ is related to the uplink rate demand $R_i^u$, as $R_i^d = {\theta _i}R_i^u$, see \cite{Sungyeon}. Applying the downlink-uplink rate ratio, the dimension of the achievable rate region degrades to 2. In this section, for simplicity, we plot the achievable rate regions of the rate pair $(R_1^u,R_2^u)$ and assume each node has equal transmission power which is ${P_j}=10, j \in \{ 1,R1,B,R2,2\}$ .

Notice that, the achievable rate regions of the four-phase relaying scheme are closed-form and independent from the time ratios as proposition 3 and 4 show. However, the achievable rate region of two-phase AF relaying scheme is related to the superposition ratio as proposition 1 shows. And the achievable rate region of two-phase DF relaying scheme is related to the time ratio and superposition ratio as proposition 2 shows. Then we need to optimize $\tau$  and $\alpha$  to get the envelope of the achievable rate regions for two-phase scheme.


In the optimization, we fix the rate $R_1^u$ of user 1 to a set value $r_1$ and find the maximal rate $R_2^u$ of user 2. The parameter $r_1$ describes the feasible rate values for user 1, ${r_1} \in \left[ {0, R_{1\max }^u} \right]$ where $R_{1\max }^u$ is determine via another optimization.
The achievable rate regions of the two-phase scheme can be obtained by a two-step optimization. Here we only present the optimization method for DF case, the AF case follows the similar way. The first step is to find $R_{1\max }^u = \max \left\{ {R_1^u} \right\}$, as shown below:
\begin{equation}
\begin{gathered}
  \mathop {\max}\limits_{0 \leqslant \tau  \leqslant 1,0 \leqslant \alpha  \leqslant 1}    ~~R_1^u,~~~~~~
  s.t.~ (\ref{defineD})~{\rm{and}}~{R_2^u} = 0,~{\rm{given}}~{\theta _1},{\theta _2}.
\label{eq1}
\end{gathered}
\end{equation}

Then, finding the achievable rate region is equivalently to solve the following optimization problem for each feasible rate point ${r_1} \in \left[ {0, R_{1\max }^u} \right]$ :
\begin{equation}
\begin{gathered}
  \mathop {\max}\limits_{0 \leqslant \tau  \leqslant 1,0 \leqslant \alpha  \leqslant 1}    ~~{R_2^u},~~~~~~
  s.t.~(\ref{defineD})~{\rm{and}}~{R_1^u} = {r_1},~{\rm{given}}~{\theta _1},{\theta _2}.
\label{eq2}
\end{gathered}
\end{equation}

The formulas (\ref{eq1}) and  (\ref{eq2}) are two conventional constrained nonlinear optimization problems for which a numerical solution for the achievable rate region can be easily found.

\section{Numerical results}\label{numerical}
The achievable rate regions of rate pair $(R_1^u,R_2^u)$ are presented in Fig.~\ref{st1}-Fig.~\ref{st6}. S2 stands for two-phase relaying scheme, S4 stands for four-phase relaying scheme. Here we only deal with six typical cases. The six typical cases are divided into two groups of three cases, where the groups have symmetric and asymmetric downlink-uplink ratio, respectively.

When the data rate is symmetric, ${\theta _1} = {\theta _2} = 1$, we focus on the performance under different SNR  conditions as Fig.~\ref{st1}-Fig.~\ref{st3} show. We consider three cases. The first case is when the links have the same SNR, that is ${\left| {{h_{11}}} \right|^2} = {\left| {{h_{12}}} \right|^2} = {\left| {{h_{22}}} \right|^2} = {\left| {{h_{21}}} \right|^2} = 1$. The second case is when the links that are direct to the BS have a lower SNR than the links that are direct to the users, that is ${\left| {{h_{11}}} \right|^2} = {\left| {{h_{21}}} \right|^2} = 1,{\left| {{h_{12}}} \right|^2} = {\left| {{h_{22}}} \right|^2} =0.1$. The third case is when the links that are direct to the BS have a higher SNR than the links  that are direct to the users, that is ${\left| {{h_{11}}} \right|^2} = {\left| {{h_{21}}} \right|^2} =0.1,{\left| {{h_{12}}} \right|^2} = {\left| {{h_{22}}} \right|^2} =1$.

When the data rate is asymmetric, ${\theta _1} \ne 1,{\theta _2} \ne 1$, we fix the SNR as ${\left| {{h_{11}}} \right|^2} = {\left| {{h_{12}}} \right|^2} = {\left| {{h_{22}}} \right|^2} = {\left| {{h_{21}}} \right|^2} = 1$ and focus on the performance under different downlink-uplink rate ratio configurations as Fig.~\ref{st4}-Fig.~\ref{st6} show. We again consider three cases.  The first case is when the downlink date rate is smaller than the uplink date rate for both U$1$ and U$2$, that is ${\theta _1}={\theta _2} =0.5$. The second case is when the downlink date rate is larger than the uplink date rate for both U$1$and U$2$, that is ${\theta _1} = {\theta _2} = 2$. The third case is the downlink data rate is larger than the uplink date rate for U$1$ and the downlink data rate is smaller than the uplink date rate for U$2$, that is ${\theta _1}{\rm{ = 2,}}~{\theta _2} = 0.5$.

For the same relaying scheme (S2 or S4), the rate regions of DF are always larger AF. The main reason is that, in DF, the noise is removed at the relay stations and not forwarded.

Consider the points in the axes where $R_1^u = 0$ or $R_2^u = 0$. Only one user communicates with BS. Therefore the four-way communication degrades to two-way communication. This explains why two-phase DF relaying (S2) and four-phase DF relaying (S4) coincide on the axes. However, there is a MA process to BS in the two-phase DF relaying scheme, that the four-phase DF relaying scheme does not have. Because the MA process can enlarge the sum rate, $\max \left\{ {R_1^u + R_2^u} \right\}$ of the two-phase DF relaying scheme is larger than the four-phase DF relaying scheme. Note that, the shape of the rate region must be convex. Two-phase DF relaying scheme and four-phase DF relaying scheme have the same axes points of the rate regions. However, two-phase DF relaying scheme has larger $\max \left\{ {R_1^u + R_2^u} \right\}$ than the four-phase DF relaying scheme. Conclude the above factors, the achievable rate region of two-phase DF relaying scheme is larger than the achievable rate region of four-phase DF relaying scheme.

We now consider AF relaying and focus on the point $R_2^u = 0$ for instance. For two-phase AF relaying, although BS and U$2$ has no information to communicate, RS$2$ also receives the information from BS which intends to RS$1$. In the second phase, RS$2$ forwards the information received before, along with the noise. Note that this noise will be present, amplified and forwarded even when there is no communication between BS and U$2$. This explains why the maximal value of the rate $R_1^u$ at $R_2^u = 0$  for two-phase AF relaying is lower than the one for four-phase AF relaying. It should be noted that, in some cases, the values in $(R_1^u,R_2^u)$ do not correspond to the capacities of the individual links, because $(R_1^d,R_2^d)$ can limit $(R_1^u,R_2^u)$ through downlink-uplink ratio. For the above situation, the axes points of two-phase AF relaying and four-phase AF relaying can be the same as Fig.~\ref{st2}, Fig.~\ref{st5} and Fig.~\ref{st6} show.

In two-phase AF relaying there are more MA processes activated simultaneously than four-phase AF relaying. And the MA process can enlarge the sum rate. This implies that, in most cases the sum-rate of the two-phase AF relaying scheme can be better than the four-phase AF relaying scheme, as Fig.~\ref{st3}, Fig.~\ref{st4} and Fig.~\ref{st6} show.
\begin{figure}[htbp]
  \begin{minipage}[t]{0.52\linewidth}
    \includegraphics[width=7 cm]{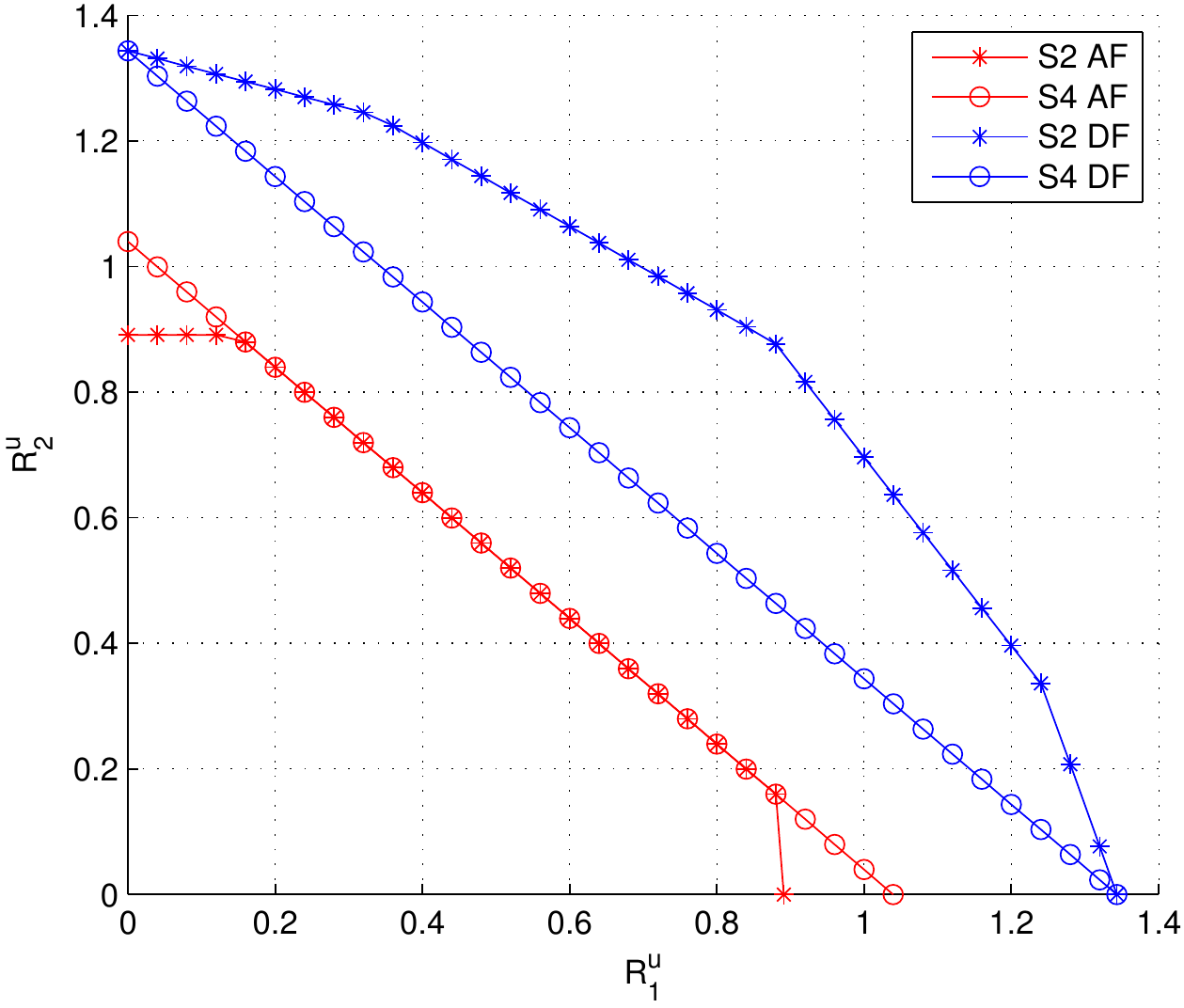}
  \end{minipage}%
  \begin{minipage}[t]{0.52\linewidth}
    \includegraphics[width=7 cm]{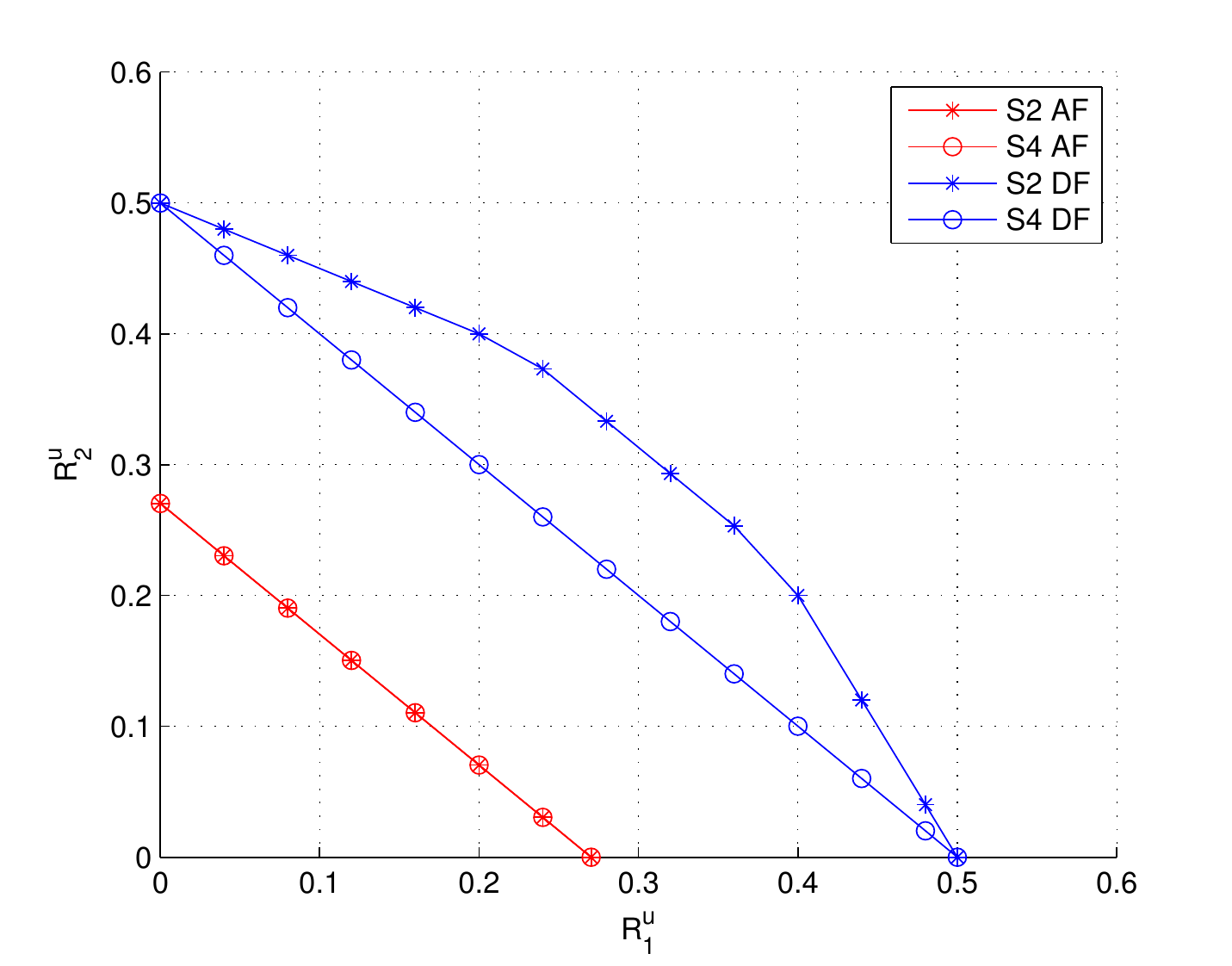}
  \end{minipage}\\[-10pt]
    \begin{minipage}[t]{.42\linewidth}
   \caption{Achievable rate region for the pair of
uplink rates, assuming a downlink-uplink ratio of ${\theta _1}{\rm{ = }}{\theta _2}{\rm{ = 1}}$ and channel values of ${\left| {{h_{11}}} \right|^2} = {\left| {{h_{12}}} \right|^2} = {\left| {{h_{22}}} \right|^2} = {\left| {{h_{21}}} \right|^2} = 1$.}
     \label{st1}
    \end{minipage}%
  \hfill%
    \begin{minipage}[t]{.42\linewidth}
      \caption{Achievable rate region for the pair of
uplink rates, assuming a downlink-uplink ratio of ${\theta _1}{\rm{ = }}{\theta _2}{\rm{ = 1}}$ and channel values of ${\left| {{h_{11}}} \right|^2} = {\left| {{h_{21}}} \right|^2} = 1,{\left| {{h_{12}}} \right|^2} = {\left| {{h_{22}}} \right|^2} =0.1$.}
       \label{st2}
    \end{minipage}%
\end{figure}
\begin{figure}[htpb]
  \begin{minipage}[t]{0.52\linewidth}
    \includegraphics[width=7 cm]{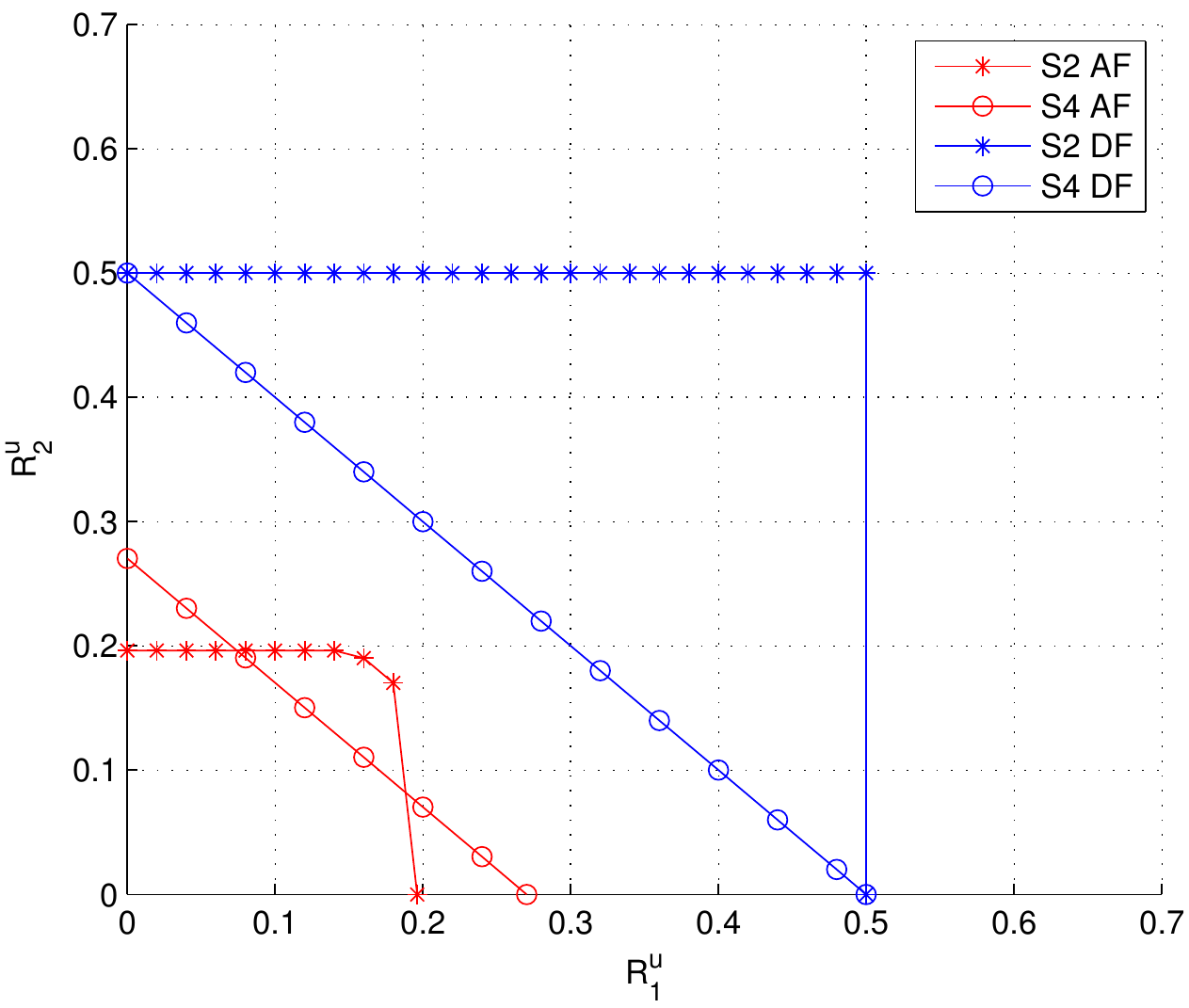}
  \end{minipage}%
  \begin{minipage}[t]{0.52\linewidth}
    \includegraphics[width=7 cm]{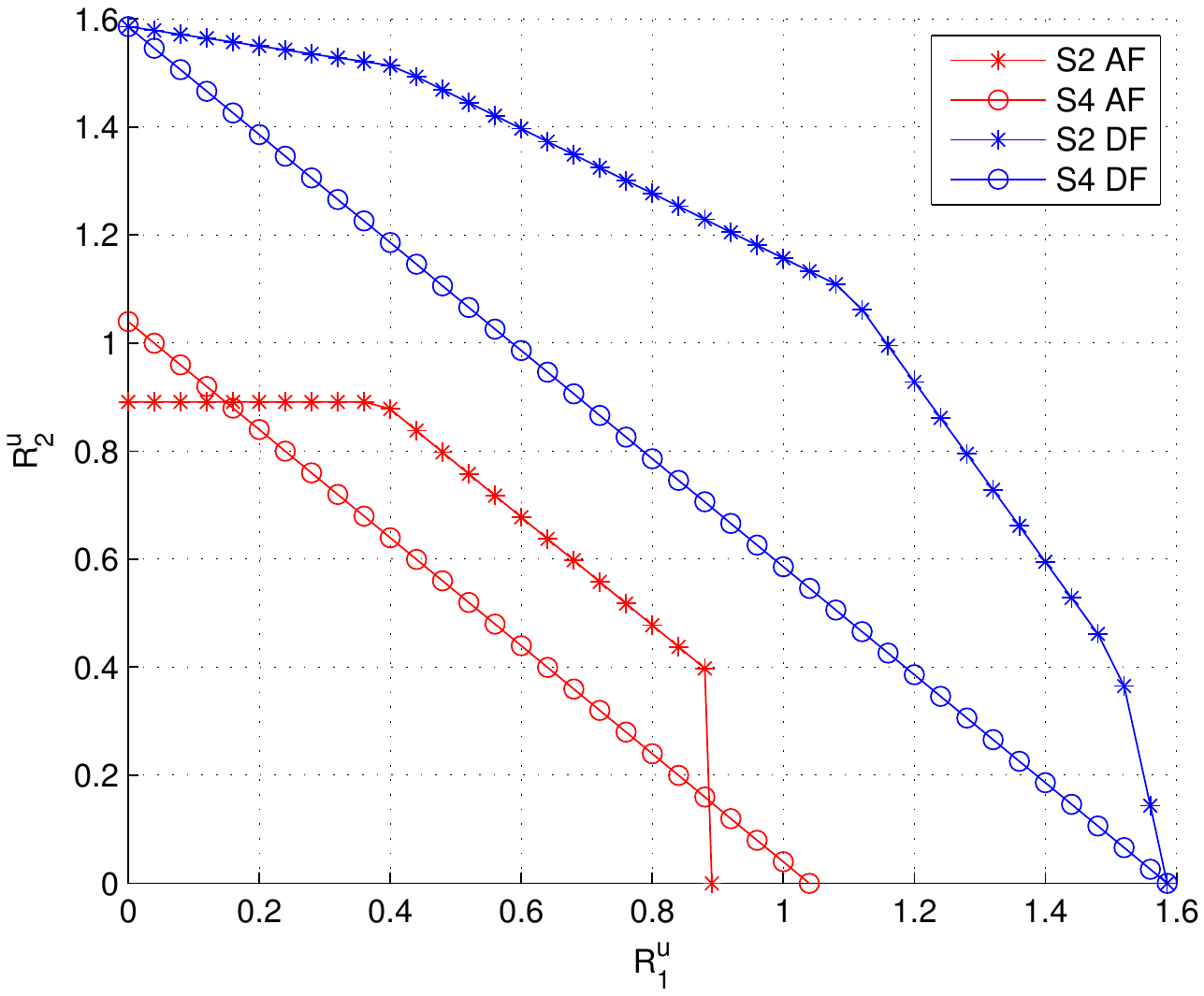}
  \end{minipage}\\[-10pt]
    \begin{minipage}[t]{.42\linewidth}
   \caption{Achievable rate region for the pair of
uplink rates, assuming a downlink-uplink ratio of ${\theta _1}{\rm{ = }}{\theta _2}{\rm{ = 1}}$ and channel values of ${\left| {{h_{11}}} \right|^2} = {\left| {{h_{21}}} \right|^2} = 0.1,{\left| {{h_{12}}} \right|^2} = {\left| {{h_{22}}} \right|^2} = 1$. }
       \label{st3}
    \end{minipage}%
  \hfill%
    \begin{minipage}[t]{.42\linewidth}
      \caption{Achievable rate region for the pair of
uplink rates, assuming a downlink-uplink ratio of ${\theta _1}{\rm{ = }}{\theta _2}{\rm{ = 0}}{\rm{.5}}$ and channel values of  ${\left| {{h_{11}}} \right|^2} = {\left| {{h_{12}}} \right|^2} = {\left| {{h_{22}}} \right|^2} = {\left| {{h_{21}}} \right|^2} = 1$.}
         \label{st4}
    \end{minipage}%
\end{figure}
\begin{figure}[htpb]
  \begin{minipage}[t]{0.52\linewidth}
    \includegraphics[width=7 cm]{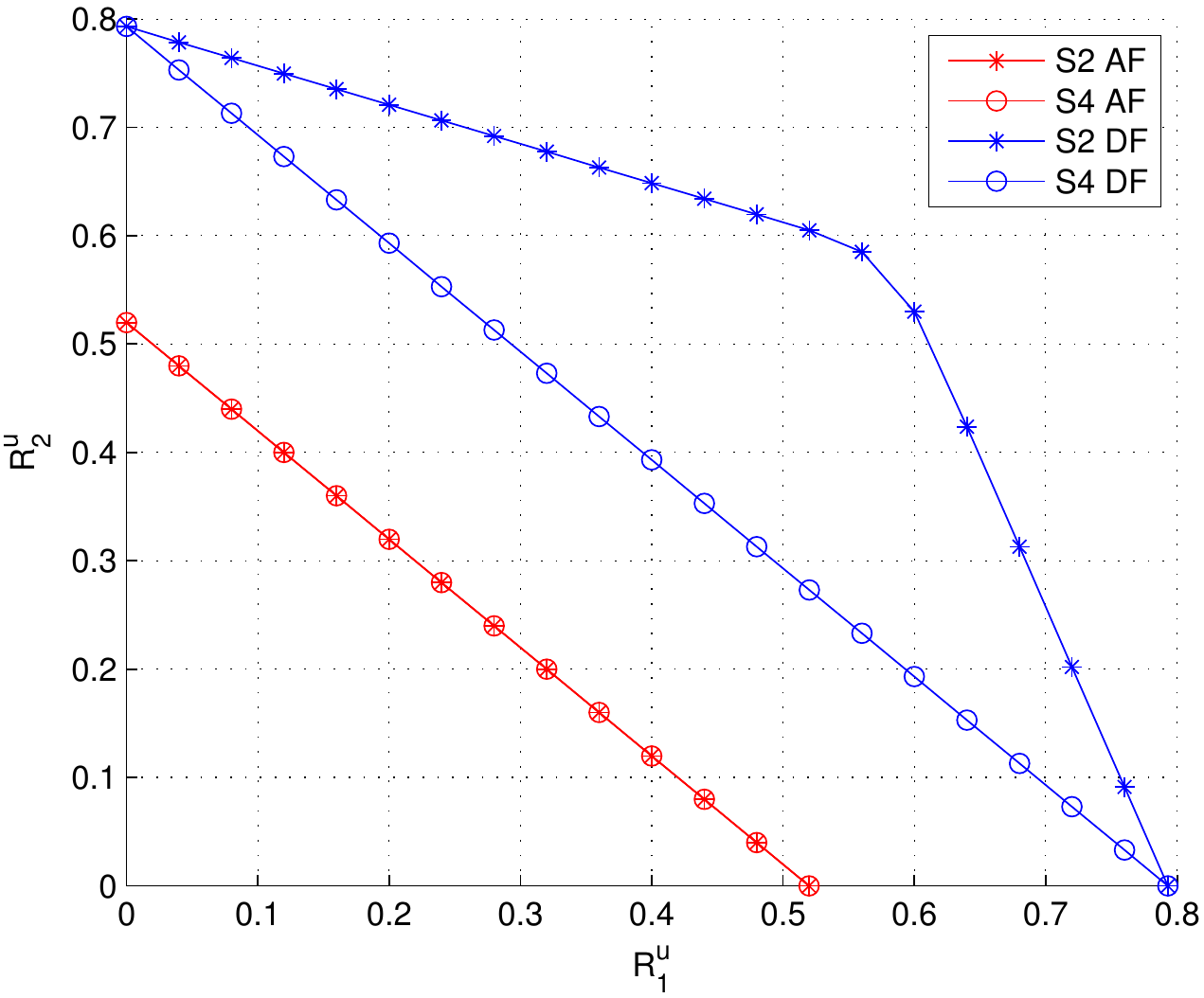}
  \end{minipage}%
  \begin{minipage}[t]{0.52\linewidth}
    \includegraphics[width=7 cm]{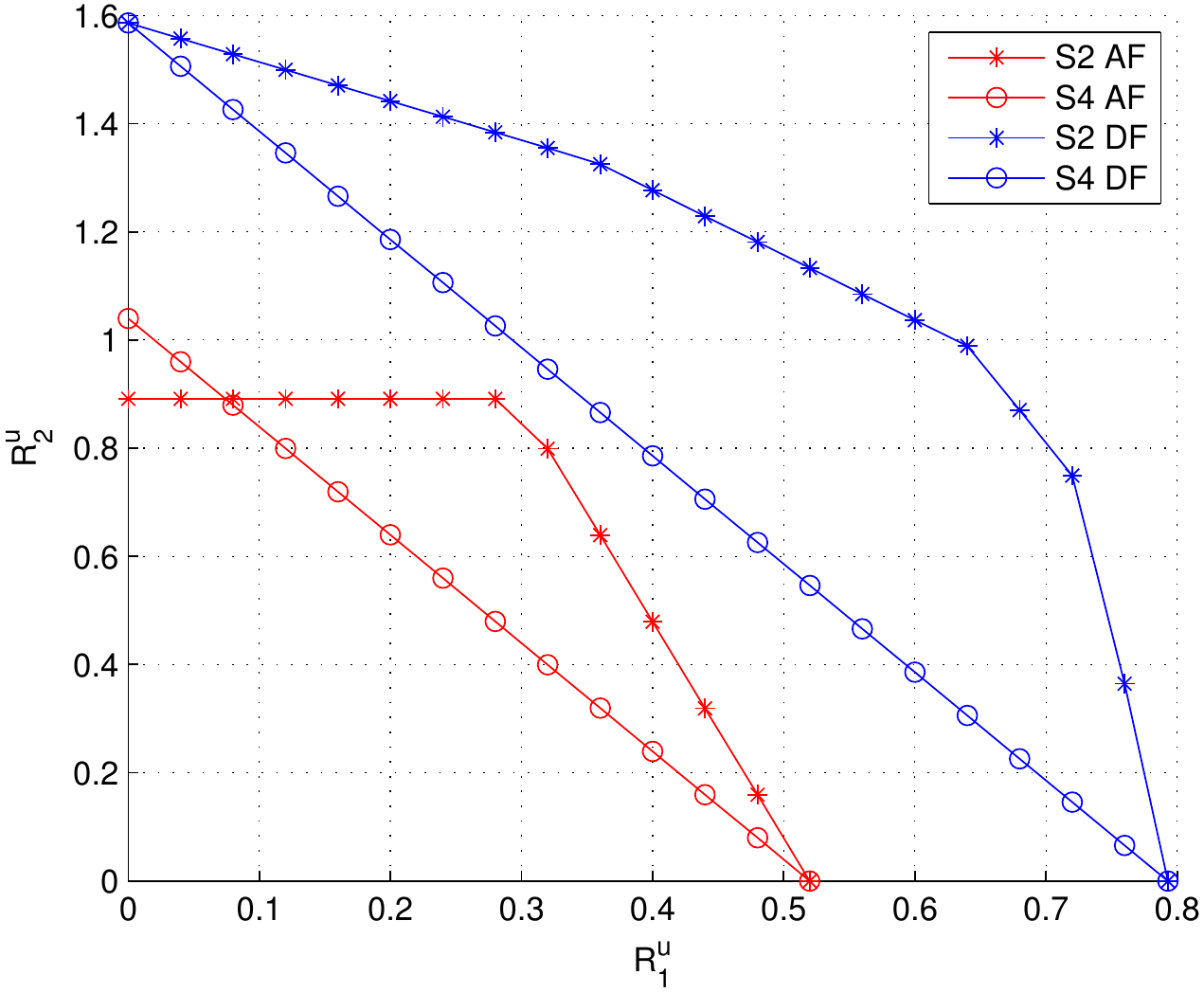}
  \end{minipage}\\[-10pt]
    \begin{minipage}[t]{.42\linewidth}
   \caption{Achievable rate region for the pair of
uplink rates, assuming a downlink-uplink ratio of ${\theta _1}{\rm{ = }}{\theta _2}{\rm{ = 2}}$ and channel values of ${\left| {{h_{11}}} \right|^2} = {\left| {{h_{12}}} \right|^2} = {\left| {{h_{22}}} \right|^2} = {\left| {{h_{21}}} \right|^2} = 1$.}
     \label{st5}
    \end{minipage}%
  \hfill%
    \begin{minipage}[t]{.42\linewidth}
      \caption{Achievable rate region for the pair of
uplink rates, assuming a downlink-uplink ratio of ${\theta _1}{\rm{ = 2,}}{\theta _2} = 0.5$ and channel values of ${\left| {{h_{11}}} \right|^2} = {\left| {{h_{12}}} \right|^2} = {\left| {{h_{22}}} \right|^2} = {\left| {{h_{21}}} \right|^2} = 1$.}
          \label{st6}
    \end{minipage}%
\end{figure}
\section{Conclusion}
We have described a new multi-way relay scenario of practical relevance in wireless cellular networks, termed four-way relaying, in which each of the two Mobile Stations (MSs) has a two-way connection to the same Base Station (BS), while each connection is through a dedicated Relay Station (RS). One of the main assumptions is that the RSs are antipodal, i. e. deployed at the ``opposite side'' of the BS, such that they are not interfering. We have proposed novel communication schemes which leverage on the ideas of wireless network coding, but are designed for the considered four-way scenario. The key idea is to have coordinated transmissions to/from the two RSs. The principle can be applied if the RS uses AF or DF relaying. The proposed communication scheme consists of two phases: broadcast and a multiple access, respectively. We compare the performance with a state-of-the-art reference scheme, time sharing is used between the two MSs, while each MS is served through a two-way relaying scheme. The results indicate that, when the RS operates in a DF mode, the achievable rate regions are significantly enlarged. On the other hand, for AF relaying, the gains are rather modest.

An interesting issue for future work is to consider other types of operation for the relay, such as physical-layer network coding, use of lattices and noisy network coding. As another line of research would be to investigate how the Base Stations and the relays should be deployed, assuming that they can use the proposed four-way relaying method.

\appendices{}

\section{}
The expression (\ref{bcop}) is an upper bound of the achievable rates of phase 2. Here we prove that this upper bound can be achieved. In other words, $R_i^d$ can achieve the capacity of $h_{i1}$ and $R_i^u$ can achieve the whole MA region in Fig.~\ref{MAregion}.  A and B are the corner points of the MA region corresponding to $(C({{{{\left| {{h_{12}}} \right|}^2}{P_{R1}}} \mathord{\left/
 {\vphantom {{{{\left| {{h_{12}}} \right|}^2}{P_{R1}}} {({{\left| {{h_{22}}} \right|}^2}{P_{R2}} + 1)}}} \right.
 \kern-\nulldelimiterspace} {({{\left| {{h_{22}}} \right|}^2}{P_{R2}} + 1)}}), C({\left| {{h_{22}}} \right|^2}{P_{R2}}))$ and
$(C({\left| {{h_{12}}} \right|^2}{P_{R1}}),C({{{{{\left| {{h_{22}}} \right|}^2}{P_{R2}}} \mathord{\left/
 {\vphantom {{{{\left| {{h_{22}}} \right|}^2}{P_{R2}}} {(\left| {{h_{12}}} \right|}}} \right.
 \kern-\nulldelimiterspace} {(\left| {{h_{12}}} \right|}}^2}$ ${P_{R1}} + 1)))$respectively.

The main idea to prove the achievability of the MA region is using one point-to-point codeword at each RS to achieve the corner points (A and B) of the MA region.  Then we apply the time-sharing between the point-to-point codeword achieving corner point A and the point-to-point codeword achieving corner point B, thereby proving that the MA region is achievable.
\begin{figure}[h]
\centering
\includegraphics[width=10 cm]{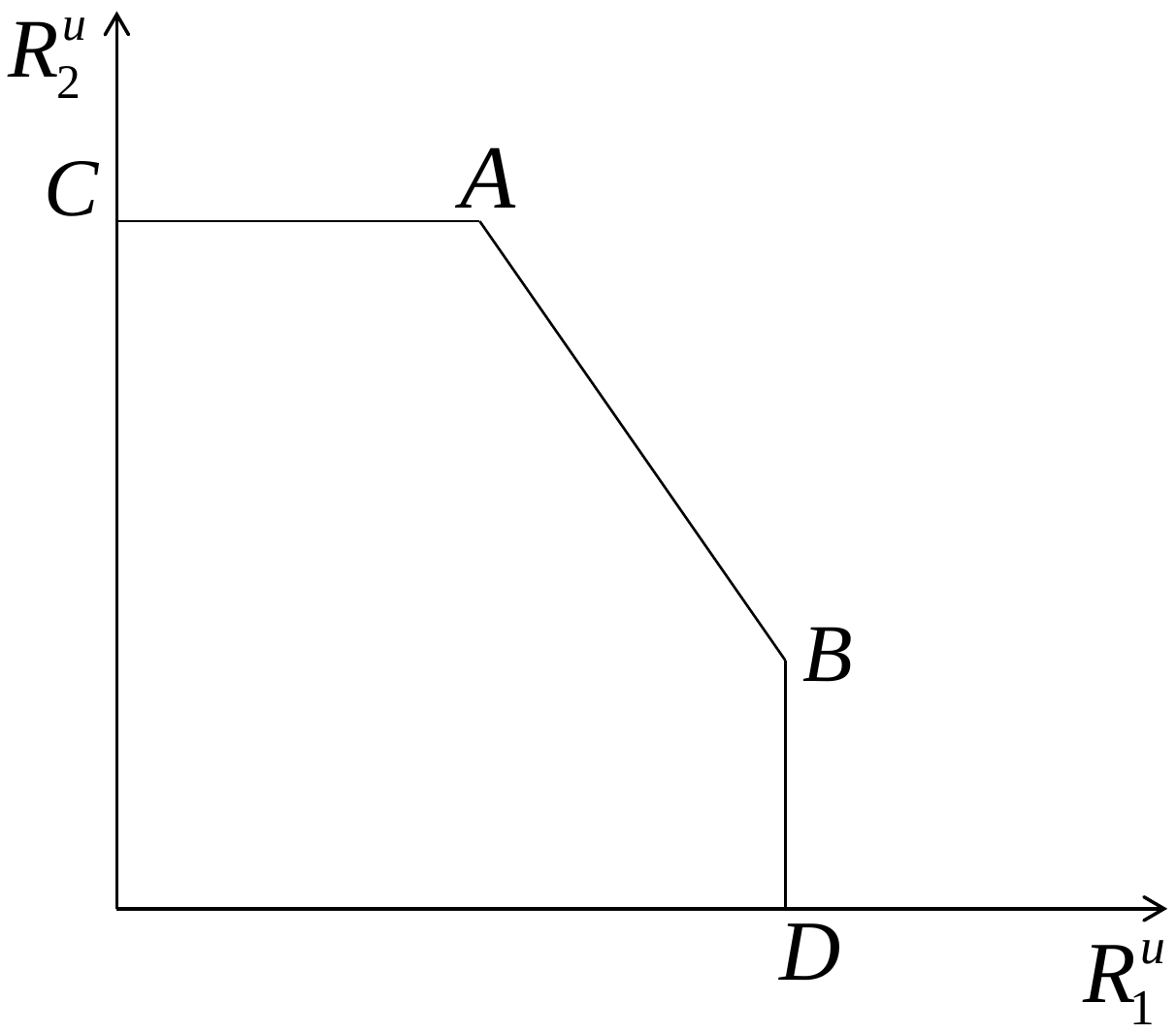}
\caption{The capacity region of a multiple access (MA) channel.}%
\label{MAregion}%
\end{figure}
For each RS, the transmission in phase 2 is a BC process with side information at the receivers. On the other hand, the simultaneous transmission of RS$1$ and RS$2$ are two components of the MA channel, defined at the BS as a receiver. This makes the two BC transmissions interrelated in terms of the achievable rates. In the following we will prove that MA region corner point A can be achieved and $R_i^d$ can achieve the capacity of $h_{i1}$ by the re-encoded codewords $x_{Ri}^A$ at RS$i$ .

When decoding, BS first treats $x_{R2}^A$ as noise in order to decode $x_{R1}^A$. Then for U$1$, RS$1$ and BS, this is a BC process with side information at the receivers. From \cite[Theorem 2]{KramerBC}, $R_1^d$ can achieve the capacity of $h_{11}$ which is $C( {\left| {{h_{11}}} \right|^2}{P_{R1}})$ (denote as $C_{11}$) and $R_1^u$ can achieve $C({{{{\left| {{h_{12}}} \right|}^2}{P_{R1}}} \mathord{\left/
 {\vphantom {{{{\left| {{h_{12}}} \right|}^2}{P_{R1}}} {({{\left| {{h_{22}}} \right|}^2}{P_{R2}} + 1)}}} \right.
 \kern-\nulldelimiterspace} {({{\left| {{h_{22}}} \right|}^2}{P_{R2}} + 1)}})$ (denote as ${R_1^u(A)}$). After the decoding of $x_{R1}^A$, BS cancels $x_{R1}^A$ from the receiving signal. Then for U$2$, RS$2$ and BS, this is a BC process with side information at the receivers. From \cite[Theorem 2]{KramerBC}, $R_2^u$ can achieve $C({\left| {{h_{22}}} \right|^2}{P_{R2}})$ (denote as ${R_2^u(A)}$) and $R_2^d$ can achieve $C( {\left| {{h_{21}}} \right|^2}{P_{R2}})$ (denote as $C_{21}$). By far, we proved that MA region corner point A can be achieved and $R_i^d$ can achieve the capacity of $h_{i1}$ by the re-encoded codewords $x_{Ri}^A$ at RS$i$. In a similar manner we can prove that the rates corresponding to the corner point B can be achieved.

To summarize, the codeword $x_{Ri}^A$ achieves $\left( {R_1^u(A), R_2^u(A), {C_{11}},{C_{21}}} \right)$ for rate tuple $(R_1^u,R_2^u,R_1^d,$ $R_2^d)$ and the codeword $x_{Ri}^B$ achieves $\left( {R_1^u(B), R_2^u(B), {C_{11}},{C_{21}}} \right)$ for  rate tuple $\left( {R_1^u, R_2^u, R_1^d, R_2^d} \right)$.
In order to achieve the rates on the line AB in Fig.~\ref{MAregion}, we can apply time-sharing between the codewords $x_{Ri}^A$ and $x_{Ri}^B$. Since U$1$, U$2$ and BS know the exact duration of $x_{Ri}^A$ and $x_{Ri}^B$ in the time-sharing codeword, the decoding process during $t$ is to decode $x_{Ri}^A$, the decoding process during $1-t$ is to decode $x_{Ri}^B$. Therefore, $\left( {tR_1^u(A) + \bar tR_1^u(B), tR_2^u(A) + \bar tR_2^u(B), {C_{11}},{C_{21}}} \right)$ is achievable for every $t \in [0,1]$,~$\bar t = 1 - t$. This proves that the line AB in Fig.~\ref{MAregion}, is achievable while $R_i^d$ can achieve the single-user capacity of the channel with $h_{i1}$. It should be noted that the time-sharing does not affect the decoding at U$i$, since U$i$ knows the detailed structure of the re-encoded codeword, such that it can apply the appropriate side information upon decoding.

Similarly, by time-sharing between corner points and axes points, the rates on the lines CA and BD in Fig.~\ref{MAregion} are also achievable. Then the whole MA region in Fig.~\ref{MAregion} is achievable while $R_i^d$ achieves the capacity of $h_{i1}$.  Summarizing the above results, we get (\ref{bcop}). The scaling factor $ 1 - \tau $ corresponds to the duration of phase 2.

\bibliographystyle{IEEEtran}
\bibliography{hua}

\end{document}